\documentclass[11pt]{article}

\usepackage{fullpage,wrapfig}
\usepackage{amsmath, amsthm, amssymb, algorithm,thmtools,algpseudocode,enumerate,caption,framed}
\usepackage{paralist, sidecap}

\usepackage[usenames,dvipsnames,svgnames,table]{xcolor}
\definecolor{darkgreen}{rgb}{0.0,0,0.9}

\usepackage[colorlinks=true,pdfpagemode=UseNone,citecolor=Blue,linkcolor=Red,urlcolor=Black,pagebackref]{hyperref}

\usepackage{tikz}

\usepackage[T1]{fontenc}
\usepackage[utf8]{inputenc}
\usepackage{authblk}

\newcommand{\mdsm}{\texttt{MDSM}}
\newcommand{\mwdsm}{\texttt{MWDSM}}
\newcommand{\mwcm}{\texttt{MWCM}}
\newcommand{\mwis}{\texttt{MWIS}}
\newcommand{\mcm}{\texttt{MCM}}

\makeatletter
\newcommand{\setword}[2]{
  \phantomsection
  #1\def\@currentlabel{\unexpanded{#1}}\label{#2}
}
\makeatother

\numberwithin{equation}{section}

\renewcommand{\vec}[1]{\mathbf{#1}}

\newtheorem{theorem}{Theorem}[section]

\newtheorem{lemma}{Lemma}[section]
\newtheorem{observation}{Observation}[section]

\newtheorem{definition}{Definition}[section]

\title{Approximating Weighted Duo-Preservation in Comparative Genomics\footnote{Appeared in proceedings of the 23rd International Computing and Combinatorics Conference (COCOON 2017)~\cite{Mehrabi2017}. This work was done when the author was at the University of Waterloo.}}

\author{Saeed Mehrabi}
\affil{{\small School of Computer Science

Carleton University, Ottawa, Canada

				\url{mehrabi235@gmail.com}}
}

\date{}

\begin{document}

\maketitle

\begin{abstract}
Motivated by comparative genomics, Chen et al.~\cite{ChenCSPWT14} introduced the Maximum Duo-preservation String Mapping (\mdsm) problem in which we are given two strings $s_1$ and $s_2$ from the same alphabet and the goal is to find a mapping $\pi$ between them so as to maximize the number of duos preserved. A \emph{duo} is any two consecutive characters in a string and it is \emph{preserved} in the mapping if its two consecutive characters in $s_1$ are mapped to same two consecutive characters in $s_2$. The \mdsm~problem is known to be \textsc{NP}-hard and there are approximation algorithms for this problem~\cite{BoriaCCPPQ16,Brubach16,DudekGO17}, but all of them consider only the ``unweighted'' version of the problem in the sense that a duo from $s_1$ is preserved by mapping to any same duo in $s_2$ regardless of their positions in the respective strings. However, it is well-desired in comparative genomics to find mappings that consider preserving duos that are ``closer'' to each other under some distance measure~\cite{FoCGBook}.

In this paper, we introduce a generalized version of the problem, called the Maximum-Weight Duo-preservation String Mapping (\mwdsm) problem that captures both duos-preservation and duos-distance measures in the sense that mapping a duo from $s_1$ to each preserved duo in $s_2$ has a weight, indicating the ``closeness'' of the two duos. The objective of the \mwdsm~problem is to find a mapping so as to maximize the total weight of preserved duos. In this paper, we give a polynomial-time 6-approximation algorithm for this problem.
\end{abstract}

\section{Introduction}
\label{sec:introduction}
Strings comparison is one of the central problems in the field of stringology with many applications such as in Data Compression and Bioinformatics. One of the most common goals of strings comparison is to measure the similarity between them, and one of the many ways in doing so is to compute the \emph{edit distance} between them. The edit distance between two strings is defined as the minimum number of edit operations to transform one string into the other. In biology, during the process of DNA sequencing for instance, computing the edit distance between the DNA molecules of different species can provide insight about the level of ``synteny'' between them; here, each edit operation is considered as a single mutation.

In the simplest form, the only edit operation that is allowed in computing the edit distance is to shift a block of characters; that is, to change the order of the characters in the string. Computing the edit distance under this operation reduces to the Minimum Common String Partition (MCSP) problem, which was introduced by Goldstein et al.~\cite{GoldsteinKZ05} (see also~\cite{SwensonMEM08}) and is defined as follows. For a string $s$, let $P(s)$ denote a partition of $s$. Given two strings $s_1$ and $s_2$ each of length $n$, where $s_2$ is a permutation of $s_1$, the objective of the MCSP problem is to find a partition $P(s_1)$ of $s_1$ and $P(s_2)$ of $s_2$ of minimum cardinality such that $P(s_2)$ is a permutation of $P(s_1)$. The problem is known to be \textsc{NP}-hard and even \textsc{APX}-hard~\cite{GoldsteinKZ05}. 

Recently, Chen et al.~\cite{ChenCSPWT14} introduced a maximization version of the MCSP problem, called the \emph{Maximum Duo-preservation String Mapping} (\mdsm) problem. A \emph{duo} in a string $s$ is a pair of consecutive characters in $s$. For two strings $s_1$ and $s_2$, where $s_2$ is a permutation of $s_1$ under a mapping $\pi$, we say that a \emph{duo is preserved} in the mapping $\pi$, if its two consecutive characters are mapped to same two consecutive characters in $s_2$. Notice that if partitions $P(s_1)$ and $P(s_2)$ are a solution of size $r$ for an instance of the MCSP problem, then this solution can be used to obtain a mapping $\pi$ from $s_1$ to $s_2$ that preserve $n-r$ duos. As such, given two strings $s_1$ and $s_2$, the objective of the \mdsm~problem is to compute a mapping $\pi$ from $s_1$ to $s_2$ that preserves a maximum number of duos. See Figure~\ref{fig:duoPreservation} for an example.
\begin{figure}[t]
\centering
\includegraphics[width=0.90\textwidth]{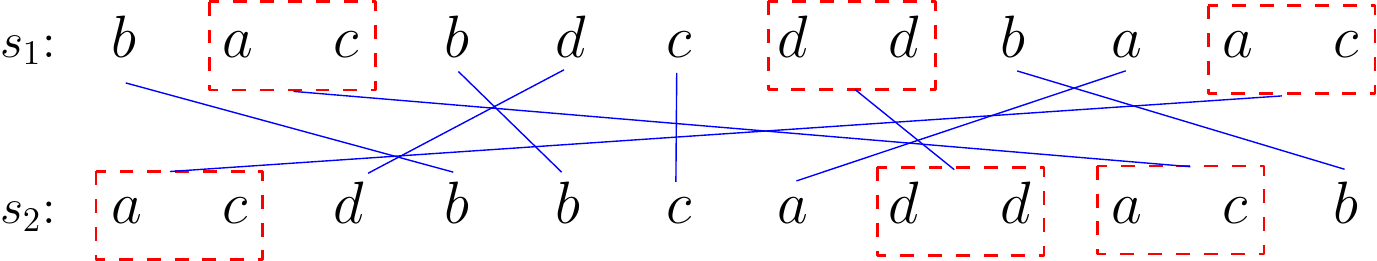}
\caption{An instance of the \mdsm~problem in which the mapping $\pi$ preserves three duos.}
\label{fig:duoPreservation}
\end{figure}

\paragraph{{\bf Related Work.}} Since the MCSP problem is \textsc{NP}-hard~\cite{GoldsteinKZ05}, there has been many works on designing polynomial-time approximation algorithms for this problem~\cite{ChenZFNZLJ05,ChrobakKS04,CormodeM07,GoldsteinKZ05,KolmanW07}. The best approximation results thus far are an $O(\log n \log^* n)$-approximation algorithm for the general version of the problem~\cite{CormodeM07}, and an $O(k)$-approximation for the $k$-MCSP problem~\cite{CormodeM07} (the $k$-MCSP is a variant of the problem in which each character can appear at most $k$ times in each string). In terms of the parameterized complexity, the problem is known to be fixed-parameter tractable with respect to $k$ and the size of an optimal partition~\cite{BulteauFKR13,JiangZZZ12}, as well as the size of an optimal solution only~\cite{BulteauK14}. For the \mdsm~problem, we observe that since the MCSP problem is \textsc{NP}-hard~\cite{GoldsteinKZ05}, the \mdsm~problem (i.e., its maximization version) is also \textsc{NP}-hard, and in fact even \textsc{APX}-hard~\cite{BoriaKLM14}. Moreover, the problem is also shown to be fixed-parameter tractable with respect to the number of duos preserved~\cite{0001CD16}. Boria et al.~\cite{BoriaKLM14} gave a 4-approximation algorithm for the \mdsm~problem, which was subsequently improved to algorithms with approximation factors of 7/2~\cite{BoriaCCPPQ16}, 3.25~\cite{Brubach16} and (recently) $(2+\epsilon)$ for any $\epsilon>0$~\cite{DudekGO17}.

\paragraph{{\bf Motivation and Problem Statement.}} Observe that in the \mdsm~problem mapping a duo from $s_1$ to a preserved duo in $s_2$ does not consider the position of the two duos in $s_1$ and $s_2$. In Figure~\ref{fig:duoPreservation}, for instance, the first $ac$ in $s_1$ is mapped to the second $ac$ in $s_2$ and the second $ac$ in $s_1$ is mapped to the first $ac$ in $s_2$. But, another (perhaps more realistic) mapping would be the one that maps the first $ac$ in $s_1$ to the first $ac$ in $s_2$ and the second one in $s_1$ to the second one in $s_2$. The latter mapping would become more desirable when solving the problem on strings of extremely long length. In fact, considering the applications of the \mdsm~problem in comparative genomics, it is much more desirable to find mappings that take into account the position of the preserved features in the two sequences~\cite{FoCGBook,Hardison03}. One reason behind this is the fact that focusing on giving priority to preserving features that are ``closer'' to each other (distance-wise under some distance measure) provides better information about the ``synteny'' of the corresponding species~\cite{FoCGBook,Hardison03}.

In this paper, we introduce a more general variant of the \mdsm~problem, called the Maximum-Weight Duo-preservation String Mapping (\mwdsm) problem. In this problem, in addition to $s_1$ and $s_2$, we are also given a \emph{weight function} defined on pairs of duos that considers the position of the duos in $s_1$ and $s_2$, and so better captures the concept of ``synteny'' in comparative genomics. Now, the objective becomes maximizing the total weight of the preserved duo (instead of maximizing the number of them). Let us define the problem more formally. For a string $s$, we denote by $D(s)$ the set of all duos in $s$ ordered from left to right. For example, if $s=acbbda$, then $D(s)=\{ac, cb, bb, bd, da\}$.
\begin{definition}[The \mwdsm~Problem]
\label{def:MWDSM}
Let $s_1$ and $s_2$ be two strings of length $n$. Moreover, let $\vec{w}:D(s_1)\times D(s_2)\rightarrow \mathbb{R^+}$ denote a weight function. Then, the \mwdsm~problem asks for a mapping $\pi$ from $s_1$ to $s_2$ that preserve a set $S$ of duos such that
\[
\sum_{d\in S}w(d, \pi(d))
\]
is maximized over all such sets $S$, where $\pi(d)$ denotes the duo in $s_2$ to which $d\in s_1$ is mapped.
\end{definition}
We note that the weight function is very flexible in the sense that it can capture any combination of duos-preservation and duos-distance measures. To our knowledge, this is the first formal study of a ``weighted version'' of the \mdsm~problem.

\paragraph{{\bf Our Result.}} Notice that the \mwdsm~problem is \textsc{NP}-hard as its unweighted variant (i.e., the \mdsm~problem) is known to be \textsc{NP}-hard~\cite{GoldsteinKZ05}. We note that the previous approximation algorithms for the \mdsm~problem do not apply to the \mwdsm~problem. In particular, both 7/2-approximation algorithm of Boria et al.~\cite{BoriaCCPPQ16} and $(2+\epsilon)$-approximation algorithm of Dudek et al.~\cite{DudekGO17} are based on the local search technique, which is known to fail for weighted problems~\cite{MustafaR10,ChanH12}. Moreover, the 3.25-approximation algorithm of Brubach~\cite{Brubach16} relies on a triplet matching approach, which involves finding a weighted matching (with specialized weights) on a particular graph, but it is not clear if the approach could handle the \mwdsm~problem with any arbitrary weight function $w$. Finally, while the linear programming algorithm of Chen et al.~\cite{ChenCSPWT14} might apply to the \mwdsm~problem, the approximation factor will likely stay the same, which is $k^2$, where $k$ is the maximum number of times each character appears in $s_1$ and $s_2$.

In this paper, we give a polynomial-time 6-approximation algorithm for the \mwdsm~problem for any arbitrary weight function $w$. To this end, we construct a vertex-weighted graph corresponding to the \mwdsm~problem and show that the problem reduces to the \emph{Maximum-Weight Independent Set (\mwis)} problem on this graph. Then, we apply the \emph{local ratio technique} to approximate the \mwis~problem on this graph. The local ratio technique was introduced by Bar-Yehuda and Even~\cite{BarYehudaE85}, and is used for designing approximation algorithms for mainly weighted optimization problems (see Section~\ref{sec:prelimins} for a formal description of this technique). While the approximation factor of our algorithm is slightly large in compare to that of algorithms for the unweighted version of the problem~\cite{BoriaCCPPQ16,Brubach16}, as we now have weights, our algorithm is much simpler as it benefits from the simplicity of the local ratio technique. To our knowledge, this is the first application of the local ratio technique to problems in stringology.

\paragraph{{\bf Organization.}} The paper is organized as follows. We first give some definitions and preliminary results in Section~\ref{sec:prelimins}. Then, we present our 6-approximation algorithm in Section~\ref{sec:approxAlg}, and will conclude the paper with a discussion on future work in Section~\ref{sec:conclusion}.

\section{Preliminaries}
\label{sec:prelimins}
In this section, we give some definitions and preliminaries. For a graph $G$, we denote the set of vertices and edges of $G$ by $V(G)$ and $E(G)$, respectively. For a vertex $u\in V(G)$, we denote the set of neighbours of $u$ by $N[u]$; note that $u\in N[u]$.

Let $\vec{w}\in \mathbb{R}^n$ be a weight vector, and let $F$ be a set of feasibility constraints on vectors $\vec{x}\in \mathbb{R}^n$. A vector $\vec{x}\in\mathbb{R}^n$ is a feasible solution to a given problem $(F, \vec{p})$ if it satisfies all of the constraints in $F$. The value of a feasible solution $\vec{x}$ is the inner product $\vec{w}\cdot\vec{x}$. A feasible solution is \emph{optimal} for a maximization (resp., minimization) problem if its value is maximal (resp., minimal) among all feasible solutions. A feasible solution $\vec{x}$ is an $\alpha$-approximation solution, or simply an $\alpha$-approximation, for a maximization problem if $\vec{w}\cdot\vec{x}\geq \alpha\cdot\vec{w}\cdot\vec{x}^*$, where $\vec{x}^*$ is an optimal solution. An algorithm is said to have an \emph{approximation factor} of $\alpha$ (or, it is called an \emph{$\alpha$-approximation algorithm}), if it always computes $\alpha$-approximation solutions.

\paragraph{{\bf Local Ratio.}} Our approximation algorithm uses the \emph{local ratio technique}. This technique was first developed by Bar-Yehuda and Even~\cite{BarYehudaE85}. Let us formally state the local ratio theorem.
\begin{theorem}~\cite{BarYehudaE85}
\label{thm:localRatio}
Let $F$ be a set of constraints, and let $\vec{w}, \vec{w}_1$ and $\vec{w}_2$ be weight vectors where $\vec{w}=\vec{w}_1+\vec{w}_2$. If $\vec{x}$ is an $\alpha$-approximation solution with respect to $(F, \vec{w}_1)$ and with respect to $(F, \vec{w}_2)$, then $\vec{x}$ is an $\alpha$-approximation solution with respect to $(F, \vec{w})$.
\end{theorem}

We now describe how the local ratio technique is usually used for solving a problem. First, the solution set is empty. The idea is to find a decomposition of the weight vector $\vec{w}$ into $\vec{w}_1$ and $\vec{w}_2$ such that $\vec{w}_1$ is an ``easy'' weight function in some sense (we will discuss this in more details later). The local ratio algorithm continues recursively on the instance $(F, \vec{w}_2)$. We assume inductively that the solution returned recursively for the instance $(F, \vec{w}_2)$ is a good approximation and need to prove that it is also a good approximation for $(F, \vec{w})$. This requires proving that the solution returned recursively for the instance $(F, \vec{w}_2)$ is also a good approximation for the instance $(F, \vec{w}_1)$. This step is usually the main part of the proof of the approximation factor.

\paragraph{{\bf Graph $G_I$.}} Given an instance of the \mwdsm~problem, we first construct a bipartite graph $G_I=(A\cup B, E)$ as follows. The vertices in the left-side set $A$ are the duos in $D(s_1)$ in order from top to bottom and the vertices in the right-side set $B$ are the duos in $D(s_2)$ in order from top to bottom. There exists an edge between two vertices if and only if they represent the same duo. See Figure~\ref{fig:gExample} for an example. Boria et al~\cite{BoriaKLM14} showed that the \mdsm~problem on $s_1$ and $s_2$ reduces to the Maximum Constrained Matching (\mcm) problem on $G_I$, which is defined as follows. Let $A=a_1,\dots,a_n$ and $B=b_1,\dots,b_n$. Then, we are interested in computing a maximum-cardinality matching $M$ such that if $(a_i, b_j)\in M$, then $a_{i+1}$ can be only matched to $b_{j+1}$ and $b_{j+1}$ can be only matched to $a_{i+1}$. In the following, we first assign weights to the edges of $G_I$ and will then show that a similar reduction holds between the \mwdsm~problem on $s_1$ and $s_2$, and a weighted version of the \mcm~problem on $G_I$.

\begin{figure}[t]
\centering
\includegraphics[width=0.25\textwidth]{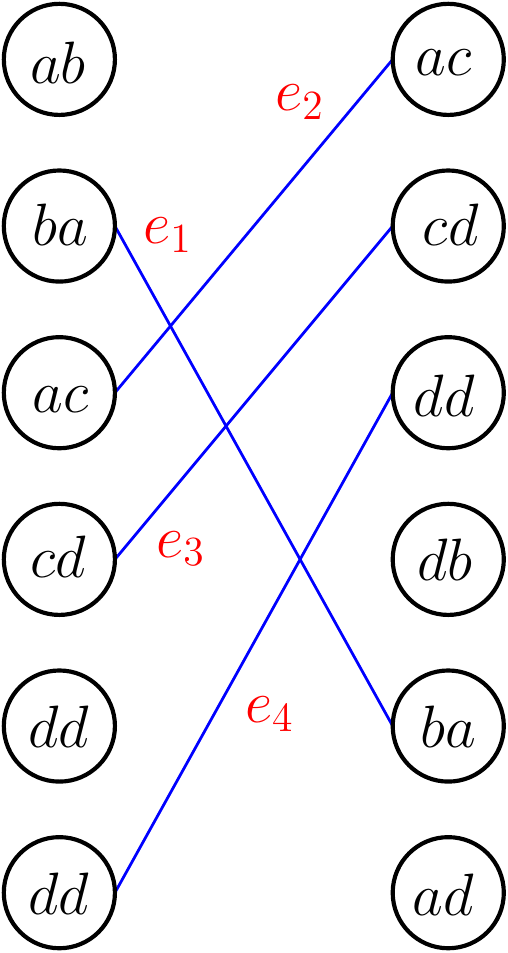}
\caption{The graph $G_I$ corresponding to $s_1=abacddd$ and $s_2=acddbad$.}
\label{fig:gExample}
\end{figure}

To weigh the edges of $G_I$, we simply assign $w(a_l, b_r)$ as the weight of $e$, for all $e\in E(G_I)$, where $a_l\in A$ and $b_r\in B$ are the endpoints of $e$ and $w(a_l, b_r)$ is given by Definition~\ref{def:MWDSM}. Now, we define the Maximum-Weight Constrained Matching (\mwcm) problem as the problem of computing a maximum-weight matching $M$ in $G_I$ such that if $(a_i, b_j)\in M$, then $a_{i+1}$ can be only matched to $b_{j+1}$ and $b_{j+1}$ can be only matched to $a_{i+1}$. To see the equivalence between the \mwdsm~problem on $s_1$ and $s_2$ and the \mwcm~problem on $G_I$, let $S$ be a feasible solution to the \mwdsm~problem with total weight $w(S)$ determined by a mapping $\pi$. Then, we can obtain a feasible solution $S'$ for the \mwcm~problem on $G_I$ by selecting the edges in $G_I$ that correspond to the preserved duos in $S$ determined by $\pi$ such that $w(S')=w(S)$. Moreover, it is not too hard to see that any feasible solution for the \mwcm~problem on $G_I$ gives a feasible solution for the \mwdsm~problem with the same weight. This gives us the following lemma.
\begin{lemma}
\label{lem:mwcmAndMWDSM}
The \mwdsm~problem on $s_1$ and $s_2$ reduces to the \mwcm~problem on $G_I$.
\end{lemma}

By Lemma~\ref{lem:mwcmAndMWDSM}, any feasible solution $M$ with total weight $w(M)$ for the \texttt{MWCM} problem on $G_I$ gives a mapping $\pi$ between the strings $s_1$ and $s_2$ that preserves a set of duos with total weight $w(M)$. As such, for the rest of this paper, we focus on the \mwcm~problem on $G_I$ and give a polynomial-time 6-approximation algorithm for this problem on $G_I$, which by Lemma~\ref{lem:mwcmAndMWDSM}, results in an approximation algorithm with the same approximation factor for the \mwdsm~problem on $s_1$ and $s_2$.

\section{Approximation Algorithm}
\label{sec:approxAlg}
In this section, we give a 6-approximation algorithm for the \mwcm~problem. We were unable to apply the local ratio directly to the \mwcm~problem on $G_I$ due to the constraint involved in the definition of the problem. Instead, we construct a vertex-weighted graph $G_C$, called the \emph{conflict graph}, and show that the \mwcm~problem on $G_I$ reduces to the Maximum-Weight Independent Set (\mwis) problem on $G_C$. We then apply the local ratio to approximate the \mwis~problem on $G_C$, which results in an approximation algorithm for the \mwcm~problem on $G_I$. Consequently, this gives us an approximation algorithm for the \mwdsm~problem on $s_1$ and $s_2$ by Lemma~\ref{lem:mwcmAndMWDSM}.

\paragraph{{\bf Graph $G_C$.}} We now describe the concept of \emph{conflict}. We say that two edges in $E(G_I)$ are \emph{conflicting} if they both cannot be in a feasible solution for the \mwcm~problem at the same time, either because they share an endpoint or their endpoints are consecutive on one side of the graph, but not on the other side. The following observation is immediate.
\begin{observation}
\label{obs:conflictingCondition}
Let $e_1=(a_i, b_j)$ and $e_2=(a_k, b_l)$ be two conflicting edges in $E(G_I)$. Then, either $k\in\{i-1, i, i+1\}$ or $l\in\{j-1, j, j+1\}$.
\end{observation}
\begin{wrapfigure}{r}{0.30\textwidth}
\centering
\vspace{-5mm}
\includegraphics[scale=0.65]{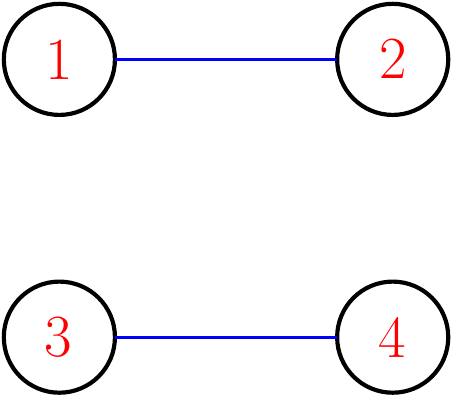}
\vspace{-5mm}
\end{wrapfigure}

We define the \emph{conflict graph} $G_C$ as follows. Let $V(G_C)$ be $E(G_I)$; that is, $V(G_C)$ is the set of all edges in $G_I$. For a vertex $i\in V(G_C)$, we denote the edge in $E(G_I)$ corresponding to $i$ by $e_i$. Two vertices $i$ and $j$ are adjacent in $G_C$ if and only if $e_i$ and $e_j$ are conflicting in $G_I$. The conflict graph $G_C$ corresponding to the graph $G_I$ in Figure~\ref{fig:gExample} is shown on the right.

To assign weights to the vertices of $G_C$, let $i$ be a vertex of $G_C$. Notice that $i$ corresponds to the edge $e_i=(a_l, b_r)$ in $E(G_I)$, where $a_l\in A$ and $b_r\in B$ are preserved duos. Then, the weight of vertex $i$ is defined as $w(i):=w(a_l, b_r)$ in which recall that $w(a_l, b_r)$ is the weight assigned to these preserved duos by Definition~\ref{def:MWDSM}. Although not precisely defined, we again note that the weight function is very flexible and it can capture any combination of duos-preservation and duos-distance measures.
\begin{lemma}
\label{lem:reductionToMIS}
The \mwcm~problem on $G_I$ reduces to the \mwis~problem on $G_C$.
\end{lemma}
\begin{proof}
Suppose that $S$ is a feasible solution to the \mwcm~problem on $G_I$ with total weight $w(S)$. For each edge $e_i\in S$, add the vertex $i\in V(G_C)$ to $S'$. Clearly, $S'$ is an independent set because two vertices $i$ and $j$ in $S'$ being adjacent would imply that $e_i$ and $e_j$ are conflicting in $G_I$, contradicting the feasibility of $S$. To see $w(S')$, notice that
\[
w(S')=\sum_{i\in S'}w(i)=\sum_{e_i=(a_l, b_r)\in S}w(a_l, b_r)=\sum_{e_i\in S}w(e_i)=w(S).
\]

Now, suppose that $S'$ is an independent set in $G_C$ with total weight $w(S')$. For each $u\in S'$, add $e_u\in E(G_I)$ to $S$. First, $S$ is a feasible solution for the \mwcm~problem on $G_I$ because the vertices of $G_C$ corresponding to any two conflicting edges in $S$ would be adjacent in $G_C$, contradicting the fact that $S'$ is an independent set. Moreover,
\[
w(S)=\sum_{e_i\in S}w(e_i)=\sum_{e_i=(a_l, b_r)\in S}w(a_l, b_r)=\sum_{i\in S'}w(i)=w(S').
\]
This completes the proof of the lemma.
\end{proof}

By Lemma~\ref{lem:reductionToMIS}, any approximation algorithm for the \mwis~problem on $G_C$ results in an approximation algorithm with the same factor for the \mwcm~problem on $G_I$. As such, for the rest of this section, we focus on the \mwis~problem on $G_C$ and show how to apply the local ratio technique to compute a 6-approximation algorithm for this problem on $G_C$.

\paragraph{{\bf Approximating the \mwis~Problem on $G_C$.}} We first formulate the \mwis~problem on $G_C$ as a linear program. We define a variable $x(u)$ for each vertex $u\in V(G_C)$; if $x(u)=1$, then vertex $u$ belongs to the independent set. The integer program assigns the binary values to the vertices with the constraint that for each clique $Q$, the sum of the values assigned to all vertices in $Q$ is at most 1.
\begin{align}
\label{aln:ip}
\text{maximize }        & \sum_{u\in V(G_C)}w(u)\cdot x(u)\\
\nonumber \text{subject to }      & \sum_{v\in Q}x(v)\leq 1 & \forall \mbox{ cliques } Q\in G_C,\\
\nonumber                         & x(u)\in\{0,1\} & \forall u\in V(G_C)
\end{align}

Note that the number of constraints in~\eqref{aln:ip} can be exponential in general, as the number of cliques in $G_C$ could be exponential. However, for the \mwis~problem on $G_C$, we can consider only a polynomial number of cliques in $G_C$. To this end, let $u=(a_i, b_j)$ be a vertex in $G_C$. By Observation~\ref{obs:conflictingCondition}, if $v=(a_k, b_l)$ is in conflict with $u=(a_i, b_j)$, then either $k\in\{i-1, i, i+1\}$ or $l\in\{j-1, j, j+1\}$. Let $S^{i-1}_u$ denote the set of all neighbours $v$ of $u$ in $G_C$ such that $v=(a_{i-1}, b_s)$ for some $b_s\in B$ (recall the bipartite graph $G_I=(A\cup B, E)$). Define $S^i_u$ and $S^{i+1}_u$ analogously. Similarly, let $S^{j-1}_u$ be the set of all neighbours $v$ of $u$ such that $v=(a_s, b_{j-1})$ for some $a_s\in A$, and define $S^j_u$ and $S^{j+1}_u$ analogously. Let $M:=\{i-1, i, i+1, j-1, j, j+1\}$. Then, by relaxing the integer constraint of the above integer program, we can formulate the \mwis~problem on $G_C$ as the following linear program.
\begin{align}
\label{aln:lp}
\text{maximize }        & \sum_{u\in V(G_C)}w(u)\cdot x(u)\\
\nonumber \text{subject to }      & \sum_{v\in S^r_u}x(v)\leq 1 & \forall u\in V(G_C), \forall r\in M,\\
\nonumber                         & x(u)\geq 0 & \forall u\in V(G_C)
\end{align}

Notice that the linear program~\eqref{aln:lp} has a polynomial number of constraints. These constraints suffice for the \mwis~problem on $G_C$ because, by Observation~\ref{obs:conflictingCondition}, the vertices $u=(a_i, b_j)$ and $v$ of $G_C$ corresponding to the two conflicting edges $e_u$ and $e_v$ in $G_I$ belong to $S^r_u$, for some $r\in M$. Moreover, we observe that any independent set in $G_C$ gives a feasible integral solution to the linear program. Therefore, the value of an optimal (not necessarily integer) solution to the linear program is an upper bound on the value of an optimal integral solution.

We are now ready to describe the algorithm. We first compute an optimal solution $\vec{x}$ for the above linear program. Then, the rounding algorithm applies a local ratio decomposition of the weight vector $\vec{w}$ with respect to $\vec{x}$. See Algorithm~\ref{alg:approxMWIS}. The key to our rounding algorithm is the following lemma.
\begin{lemma}
\label{lem:boundingNs}
Let $\vec{x}$ be a feasible solution to~\eqref{aln:lp}. Then, there exists a vertex $u\in V(G_C)$ such that
\[
\sum_{v\in N[u]}x(v)\leq 6.
\]
\end{lemma}
\begin{proof}
Let $u\in V(G_C)$. Notice that $u$ corresponds to an edge in $G_I$; that is, $u=(a_i, b_j)$, where $a_i$ and $b_j$ are a pair of preserved duos in the mapping $\pi$ from $s_1$ to $s_2$. Observe that $v\in N[u]$ for some $v=(a_k, b_l)\in V(G_C)$ if and only if $(a_k, b_l)$ conflicts with $(i, j)$ in $G_I$. Let $M:=\{i-1, i, i+1, j-1, j, j+1\}$ and define the set $S^r_u$ as above, for all $r\in M$. Note that the vertices in $S^r_u$ form a clique in $G_C$, for all $r\in M$, because the set of edges corresponding to the vertices of $S^r_u$ in $G_I$ all share one endpoint (in particular, this endpoint is in $A$ if $r\in\{i-1, i, i+1\}$ or it is in $B$ if $r\in\{j-1, j, j+1\}$). See Figure~\ref{fig:cliques} for an illustration. This means by the first constraint of the linear program~\eqref{aln:lp} that
\[
\sum_{v\in S^r_u}x(v)\leq 1,
\]
for all $r\in M$. Therefore, we have
\[
\sum_{v\in N[u]}x(v)\leq\sum_{r\in M}\sum_{v\in S^r_u}x(v)=|M|=6.
\]
This completes the proof of the lemma.
\end{proof}

\begin{algorithm}[t]
\caption{\textsc{ApproximateMWIS}($G_C$)}
\label{alg:approxMWIS}
\begin{algorithmic}[1]
\State Delete all vertices with non-positive weight. If no vertex remains, then return the empty set;
\State Let $u\in V(G_C)$ be a vertex satisfying
\[
\sum_{v\in N[u]}x(v)\leq 6.
\]
Then, decompose $\vec{w}$ into $\vec{w}:=\vec{w_1}+\vec{w_2}$ as follows:
\[
w_1(v):=
     \begin{cases}
       w(u) &\quad\text{if } v\in N[u],\\
       0 &\quad\text{otherwise.}\\
     \end{cases}
\]
\State Solve the problem recursively using $\vec{w_2}$ as the weight vector. Let $S'$ be the independent set returned;
\State If $u$ is not adjacent with some vertex in $S'$, then return $S:=S'\cup\{u\}$; otherwise, return $S:=S'$.
\end{algorithmic}
\end{algorithm}

We now analyze the algorithm. First, the set $S$ returned by the algorithm is clearly an independent set. The following lemma establishes the approximation factor of the algorithm.
\begin{lemma}
\label{lem:approxFactor}
Let $\vec{x}$ be a feasible solution to~\eqref{aln:lp}. Then, $w(S)\geq \frac{1}{6}(\vec{w}\cdot\vec{x})$.
\end{lemma}
\begin{proof}
We prove the lemma by induction on the number of recursive calls. In the base case, the set returned by the algorithm satisfies the lemma because no vertices have remained. Moreover, the first step that removes all vertices with non-positive weight cannot decrease the right-hand side of the above inequality.

We now prove the induction step. Suppose that $\vec{z}$ and $\vec{z'}$ correspond to the indicator vectors for $S$ and $S'$, respectively. By induction, $\vec{w_2}\cdot\vec{z'}\geq\frac{1}{6}(\vec{w_2}\cdot\vec{x})$. Since $w_2(u)=0$, we have $\vec{w_2}\cdot\vec{z}\geq\frac{1}{6}(\vec{w_2}\cdot\vec{x})$. From the last step of the algorithm, we know that at least one vertex from $N[u]$ is in $S$ and so we have
\[
\vec{w_1}\cdot\vec{z}=w(u)\sum_{v\in N[u]}z(v)\geq w(u).
\]
Moreover, by Lemma~\ref{lem:boundingNs},
\[
\vec{w_1}\cdot\vec{x}=w(u)\sum_{v\in N[u]}x(v)\leq 6w(u).
\]
Hence, $\vec{w_1}\cdot\vec{x}\leq 6w(u)\leq 6(\vec{w_1}\cdot\vec{z})$, which gives $\vec{w_1}\cdot\vec{z}\geq\frac{1}{6}(\vec{w_1}\cdot\vec{x})$. Therefore, we conclude that $(\vec{w_1}+\vec{w_2})\cdot\vec{z}\geq\frac{1}{6}(\vec{w_1}+\vec{w_2})\cdot\vec{x}$ and so $w(S)\geq\frac{1}{6}\vec{w}\cdot\vec{x}$. This completes the proof of the lemma.
\end{proof}

\begin{figure}[t]
\centering
\includegraphics[scale=0.65]{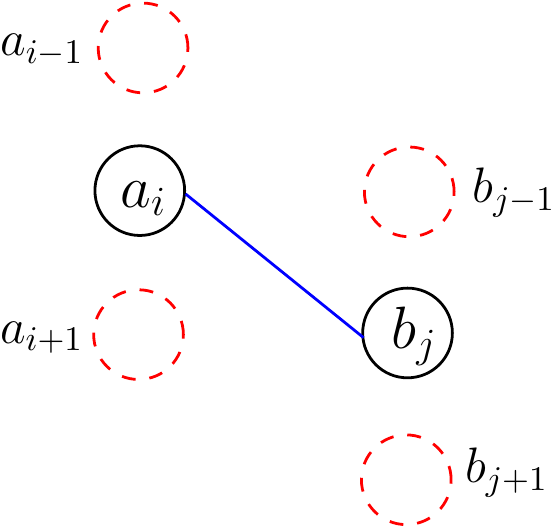}
\caption{Graph $G_I$ with edge $u=(a_i, b_j)$. The edge corresponding to any vertex $v\in N[u]$ in $G_C$ is incident to at least one of the six vertices in $\{i-1, i, i+1, j-1, j, j+1$.}
\label{fig:cliques}
\end{figure}

Since there exists at least one vertex $u$ for which $w_2(u)=0$ in each recursive step, Algorithm~\ref{alg:approxMWIS} terminates in polynomial time. Therefore, by Lemmas~\ref{lem:mwcmAndMWDSM},~\ref{lem:reductionToMIS} and~\ref{lem:approxFactor}, we have the main result of this paper.
\begin{theorem}
\label{thm:mainResult}
There exists a polynomial-time 6-approximation algorithm for the \mwdsm~problem on $s_1$ and $s_2$.
\end{theorem}

\section{Conclusion}
\label{sec:conclusion}
In this paper, we studied a weighted version of the \mdsm~problem~\cite{ChenCSPWT14} that considers the position of the preserved duos in the respective input strings (i.e., the \mwdsm~problem). This is a natural variant of the problem, as considering the position of the preserved features in the strings provides solutions with better quality in many applications, such as in comparative genomics in which more weight could indicate more synteny between the corresponding preserved features. We gave a polynomial-time 6-approximation algorithm for the \mwdsm~problem using the local ratio technique. Although the approximation factor of our algorithm is a bit large in compare to that of algorithms for the unweighted version of the problem, our algorithm is much simpler as it benefits from the simplicity of the local ratio technique. Giving approximation algorithms with better approximation factors for the \mwdsm~problem remains open for future work.

\bibliographystyle{plain}
\bibliography{ref}

\begin{thebibliography}{10}

\bibitem{BarYehudaE85}
Reuven Bar-Yehuda and Shimon Even.
\newblock A local-ratio theorem for approximating the weighted vertex cover
  problem.
\newblock In G.~Ausiello and M.~Lucertini, editors, {\em Analysis and Design of
  Algorithms for Combinatorial Problems}, volume 109, pages 27--45.
  North-Holland, 1985.

\bibitem{0001CD16}
Stefano Beretta, Mauro Castelli, and Riccardo Dondi.
\newblock Parameterized tractability of the maximum-duo preservation string
  mapping problem.
\newblock {\em Theor. Comput. Sci.}, 646:16--25, 2016.

\bibitem{BoriaCCPPQ16}
Nicolas Boria, Gianpiero Cabodi, Paolo Camurati, Marco Palena, Paolo Pasini,
  and Stefano Quer.
\newblock A 7/2-approximation algorithm for the maximum duo-preservation string
  mapping problem.
\newblock In {\em proceedings of the 27th Annual Symposium on Combinatorial
  Pattern Matching (CPM 2016), Tel Aviv, Israel}, pages 11:1--11:8, 2016.

\bibitem{BoriaKLM14}
Nicolas Boria, Adam Kurpisz, Samuli Lepp{\"{a}}nen, and Monaldo Mastrolilli.
\newblock Improved approximation for the maximum duo-preservation string
  mapping problem.
\newblock In {\em proceedings of the 14th International Workshop on Algorithms
  in Bioinformatics (WABI 2014), Wroclaw, Poland}, pages 14--25, 2014.

\bibitem{Brubach16}
Brian Brubach.
\newblock Further improvement in approximating the maximum duo-preservation
  string mapping problem.
\newblock In {\em proceedings of the 16th International Workshop on Algorithms
  in Bioinformatics (WABI 2016), Aarhus, Denmark}, pages 52--64, 2016.

\bibitem{BulteauFKR13}
Laurent Bulteau, Guillaume Fertin, Christian Komusiewicz, and Irena Rusu.
\newblock A fixed-parameter algorithm for minimum common string partition with
  few duplications.
\newblock In {\em proceedings of the 13th International Workshop on Algorithms
  in Bioinformatics (WABI 2013), Sophia Antipolis, France}, pages 244--258,
  2013.

\bibitem{BulteauK14}
Laurent Bulteau and Christian Komusiewicz.
\newblock Minimum common string partition parameterized by partition size is
  fixed-parameter tractable.
\newblock In {\em proceedings of the Twenty-Fifth Annual {ACM-SIAM} Symposium
  on Discrete Algorithms (SODA 2014), Portland, Oregon, USA}, pages 102--121,
  2014.

\bibitem{ChanH12}
Timothy~M. Chan and Sariel Har{-}Peled.
\newblock Approximation algorithms for maximum independent set of pseudo-disks.
\newblock {\em Discrete {\&} Computational Geometry}, 48(2):373--392, 2012.

\bibitem{ChenCSPWT14}
Wenbin Chen, Zhengzhang Chen, Nagiza~F. Samatova, Lingxi Peng, Jianxiong Wang,
  and Maobin Tang.
\newblock Solving the maximum duo-preservation string mapping problem with
  linear programming.
\newblock {\em Theor. Comput. Sci.}, 530:1--11, 2014.

\bibitem{ChenZFNZLJ05}
Xin Chen, Jie Zheng, Zheng Fu, Peng Nan, Yang Zhong, Stefano Lonardi, and Tao
  Jiang.
\newblock Assignment of orthologous genes via genome rearrangement.
\newblock {\em {IEEE/ACM} Trans. Comput. Biology Bioinform.}, 2(4):302--315,
  2005.

\bibitem{ChrobakKS04}
Marek Chrobak, Petr Kolman, and Jir{\'{\i}} Sgall.
\newblock The greedy algorithm for the minimum common string partition problem.
\newblock In {\em proceedings of the 7th International Workshop on
  Approximation Algorithms for Combinatorial Optimization Problems (APPROX
  2004), Cambridge, MA, USA}, pages 84--95, 2004.

\bibitem{CormodeM07}
Graham Cormode and S.~Muthukrishnan.
\newblock The string edit distance matching problem with moves.
\newblock {\em {ACM} Trans. Algorithms}, 3(1):2:1--2:19, 2007.

\bibitem{DudekGO17}
Bartlomiej Dudek, Pawel Gawrychowski, and Piotr Ostropolski{-}Nalewaja.
\newblock A family of approximation algorithms for the maximum duo-preservation
  string mapping problem.
\newblock {\em CoRR}, abs/1702.02405, 2017.

\bibitem{GoldsteinKZ05}
Avraham Goldstein, Petr Kolman, and Jie Zheng.
\newblock Minimum common string partition problem: Hardness and approximations.
\newblock {\em Electr. J. Comb.}, 12, 2005.

\bibitem{Hardison03}
Ross~C. Hardison.
\newblock Comparative genomics.
\newblock {\em PLoS Biol.}, 1(2):e58, 2003.

\bibitem{JiangZZZ12}
Haitao Jiang, Binhai Zhu, Daming Zhu, and Hong Zhu.
\newblock Minimum common string partition revisited.
\newblock {\em J. Comb. Optim.}, 23(4):519--527, 2012.

\bibitem{KolmanW07}
Petr Kolman and Tomasz Walen.
\newblock Reversal distance for strings with duplicates: Linear time
  approximation using hitting set.
\newblock {\em Electr. J. Comb.}, 14(1), 2007.

\bibitem{Mehrabi2017}
Saeed Mehrabi.
\newblock Approximating weighted duo-preservation in comparative genomics.
\newblock In {\em proceedings of the 23rd International Computing and
  Combinatorics Conference (COCOON 2017), Hong Kong, China}, pages 396--406,
  2017.

\bibitem{FoCGBook}
Arcady~R. Mushegian.
\newblock {\em Foundations of Comparative Genomics}.
\newblock Academic Press (AP), 2007.

\bibitem{MustafaR10}
Nabil~H. Mustafa and Saurabh Ray.
\newblock Improved results on geometric hitting set problems.
\newblock {\em Discrete {\&} Computational Geometry}, 44(4):883--895, 2010.

\bibitem{SwensonMEM08}
Krister~M. Swenson, Mark Marron, Joel~V. Earnest{-}DeYoung, and Bernard M.~E.
  Moret.
\newblock Approximating the true evolutionary distance between two genomes.
\newblock {\em {ACM} Journal of Experimental Algorithmics}, 12:3.5:1--3.5:17,
  2008.

\end{thebibliography}

\end{document}